\documentclass[11pt,utf8,A4paper]{article}
\usepackage[utf8]{inputenc}
\usepackage[T1]{fontenc}
\usepackage{amsmath,amssymb}
\usepackage{stmaryrd}
\usepackage{amsthm}
\usepackage{graphicx}
\usepackage{float}
\usepackage{tikz-cd}
\usepackage{txfonts}
\usepackage{hyperref}

\title{Non-locality of the meet levels of the Trotter-Weil Hierarchy}
\author{João Daniel Moreira\thanks{E-mail address: \texttt{joaodanielmoreira@gmail.com}\newline Address: (c/o Prof. Jorge Almeida) Departamento de Matemática, Faculdade de Ciências, Universidade do Porto, Rua do Campo Alegre, 687, 4169-007 Porto, Portugal}}

\theoremstyle{plain}
\newtheorem{thm}{Theorem}
\newtheorem{lem}[thm]{Lemma}
\newtheorem{prop}[thm]{Proposition}
\newtheorem*{cor}{Corollary}

\theoremstyle{definition}

\theoremstyle{remark}

\def\malcev{\mathbin{\raise1pt\hbox{\footnotesize$\bigcirc$\kern-6pt\raise1pt\hbox{\tiny$m$}\kern1pt}}}

\newcommand{\widesim}[1][1.5]{\mathrel{\scalebox{#1}[1]{$\sim$}}}
\newcommand{\widensim}[1][1.5]{\mathrel{\scalebox{#1}[1]{$\nsim$}}}

\begin{document}

\maketitle

\begin{abstract}
We prove that the meet level $m$ of the Trotter-Weil hierarchy, $\mathsf{V}_m$, is not local for all $m\geq 1$, as conjectured in a paper by Kufleitner and Lauser.
In order to show this, we explicitly provide a language whose syntactic semigroup is in $L \mathsf{V}_m$ and not in $\mathsf{V}_m*\mathsf{D}$.
\end{abstract}
 
\section{Introduction}

With the advent of computer science and its applications, the study of formal languages received a significant boost in the 1960s.
An important problem that emerged is to determine whether a given regular language has a certain type of property.
One of the first results in this direction is Sch\"utzenberger's characterization of star-free languages as those whose syntactic monoid is finite and aperiodic~\cite{Schutzenberger:1965}.
This eventually led Eilenberg~\cite{Eilenberg:1976} to formulate a general framework for the algebraic characterization of classes of languages, which came to be known as Eilenberg's correspondence.
One of the problems that received most attention is the dot-depth problem of Brzozowski~\cite{Brzozowski&Knast:1978}.
Besides the combinatorial motivation stemming from the definition of the dot-depth hierarchy, there is also an important connection with logic~\cite{Thomas:1982}, relating levels of the hierarchy with quantifier alternation.
While there has been significant progress in recent years~\cite{Place&Zeitoun:2014a,Place&Zeitoun:2014b,Almeida&Bartonova&Klima&Kunc:2015}, the dot-depth problem remains open.
Meanwhile, several hierarchies for subclasses of languages have been studied, in particular the  class of disjoint unions of unambiguous products $A_0^*a_1A_1^*\cdots a_n A_n^*$ with $a_i \in A$ and $A_i \subseteq A$, whose algebraic counterpart is the famous pseudovariety $\mathsf{DA}$~\cite{Schutzenberger:1976,Tesson&Therien:2002} and whose logic version is the quantifier alternation hierarchy in two variables~\cite{Kufleitner&Lauser:2013}.
This hierarchy had previously been considered by Trotter and Weil due to its algebraic significance~\cite{Trotter&Weil:1997}.

In fact, there are two different flavours of such hierarchies: one in which the empty word is considered with corresponding algebraic counterpart of monoids, and the other without the empty word and semigroups.
The former version of the dot-depth hierarchy is known as the Straubing-Thérien hierarchy, while the latter is known as the Brzozowski hierarchy.
In general, the basic problem is to decide whether a given regular language lies in a certain level of one of these hierarchies.
In algebraic terms, to decide whether a given semigroup or monoid lies in the corresponding pseudovariety.
At the pseudovariety level, the transition between a monoid hierarchy and the corresponding semigroup hierarchy is obtained by applying the operator $\mathsf{V}\mapsto \mathsf{V}*\mathsf{D}$, where $*$ denotes semidirect product and $\mathsf{D}$ is the pseudovariety of so-called definite semigroups, in which the idempotents are right zeros~\cite{Straubing:1985,Kufleitner&Weil:2012,Kufleitner&Lauser:2013}.
This operator has been extensively studied and a key problem proposed by Eilenberg~\cite{Eilenberg:1976} is to determine the pseudovarieties of monoids $\mathsf{V}$ such that $\mathsf{V}*\mathsf{D}=L\mathsf{V}$, which are called local.
A categorical characterization of local pseudovarieties was obtained by Tilson~\cite{Tilson:1987}, who also gave a criterion for membership in $\mathsf{V}*\mathsf{D}$, known as the Delay Theorem, which is, in a sense, a categorical reformulation of earlier work of Straubing~\cite{Straubing:1985}.
Recall that $L\mathsf{V}$ consists of the semigroups $S$ whose local monoids $eSe$, for an arbitrary idempotent $e$ of $S$.
While $L\mathsf{V}$ is obviously decidable if so is $\mathsf{V}$, for non-local pseudovarieties $\mathsf{V}$ the membership problem for $\mathsf{V}*\mathsf{D}$ is more complicated.
In the particular case of the dot-depth hierarchy and the corresponding hierarchy for $\mathsf{DA}$, decidability is known to be preserved under the operator $\mathsf{V}\mapsto\mathsf{V}*\mathsf{D}$~\cite{Straubing:1985,Kufleitner&Weil:2012,Kufleitner&Lauser:2013}.

For the monoid versions of both the dot-depth (Straubing-Thérien) hierarchy and the $\mathsf{DA}$ (Trotter-Weil) hierarchy, the first non-trivial level is the famous class of piecewise testable languages with algebraic counterpart the pseudovariety $\mathsf{J}$ of finite $\mathcal{J}$-trivial monoids~\cite{Simon:1975}.
Knast~\cite{Knast:1983a} showed that $\mathsf{J}$ is not local by giving an explicit example of a language that is an element of $L\mathsf{J}$ and does not belong to $\mathsf{J}*\mathsf{D}$.
Kufleitner and Lauser~\cite{Kufleitner&Lauser:2013} conjectured a generalization of Knast's result, which states that the $m$-th meet level of the Trotter-Weil hierarchy is not local for all $m\geq 1$.
The purpose of this paper is to establish this conjecture by explicitly exhibiting languages whose syntactic semigroups are locally in the $m$-th meet level $\mathsf{V}_m$ of the Trotter-Weil hierarchy but do not lie in $\mathsf{V}_m*\mathsf{D}$.

\section{Preliminaries}

The reader is referred to standard textbooks
\cite{Pin:1986;bk,Almeida:1994a,Rhodes&Steinberg:2009qt} for general
background and undefined terminology. In particular, since no
essential knowledge on Mal'cev or semidirect products is required,
we will not go into the details of presenting such operations on
pseudovarieties.

For a semigroup $S$, we denote by $S^I$ the monoid which is obtained from $S$ by adding a new element $1$ that multiplies as an identity.

Throughout this paper, $A$ denotes a finite alphabet.
We say that a monoid homomorphism $\varphi:A^{*}\rightarrow M$ recognizes a language $L\subseteq A^{*}$ if $L=\varphi^{-1}\left( \varphi(L)\right) $. We then also say that $M$ recognizes $L$.
A language is recognizable if it is recognized by a finite monoid. It is well known that a language $L\subseteq A^{*}$ is regular if and only if it is recognizable~\cite{Pin:1986;bk}.
An immediate consequence is that the set of regular languages is closed under complementation within the corresponding free monoid.

Let $L$ be a subset of a semigroup $S$.
The congruence $\widesim_{L}$, defined on  $S$ by $u\widesim_{L}v$ if for all $ x,y\in S^I$, $xuy\in L$ if and only if $ xvy \in L$, is called the syntactic congruence of $L$. The quotient $S/{\widesim_{L}}$, denoted by $\operatorname{Synt}(L)$, is called the syntactic semigroup of $L$.

A non-empty class $\mathsf{V}$ of finite semigroups is a pseudovariety if it is closed under taking subsemigroups, homomorphic images and finite direct products.
Pseudovarieties of monoids are defined similarly.
Let $\Sigma^{+}$ be the free semigroup over a countable alphabet $\Sigma$.
We denote by $\widehat{\Sigma^+}$ the profinite completion of $\Sigma^+$, which may be described as the inverse limit of the finite $\Sigma$-generated semigroups~\cite{Almeida:2003cshort} and whose elements are called \emph{pseudowords} over $\Sigma$.
The set $\widehat{\Sigma^+}$ has a structure of compact semigroup in which $\Sigma^+$ can be naturally viewed as a dense subsemigroup.
The essential knowledge that we require about $\widehat{\Sigma^+}$ is the following characteristic universal property: for every mapping $\varphi:\Sigma\to S$ into a finite semigroup $S$, there is a unique continuous homomorphic extension $\hat{\varphi}:\widehat{\Sigma^+}\to S$, where $S$ is endowed with the discrete topology.

Given a finite index congruence $\widesim$ on a free semigroup $A^+$, one may consider the natural homomorphism $\varphi: A^+ \to A^+/{\widesim}$ and its unique extension to a continuous homomorphism $\hat{\varphi}:\widehat{A^+}\to A^+/{\widesim}$.
The kernel congruence of $\hat{\varphi}$, which is a clopen subset of $\widehat{A^+}\times \widehat{A^+}$ and the topological closure of $\widesim$, is called the clopen extension of $\widesim$ to $\widehat{A^+}$.
In the special case of a syntactic congruence $\widesim_L$, the clopen extension of $\widesim_L$ to $\widehat{A^+}$ is the syntactic congruence $\widesim_{\overline{L}}$ of the closure $\overline{L}$. 

We say that the semigroup $S$ is \emph{equidivisible} if whenever $s_1 ,s_2 , t_1 ,t_2$ are elements of $S$ such that $s_1s_2=t_1t_2$, either $s_1=t_1$ and $s_2=t_2$, or there exists $x$ in $S$ such that $s_1x=t_1$ and $s_2=xt_2$, or such that $s_1=t_1x$ and $xs_2=t_2$.
It will be useful in the sequel to take into account that the semigroup $\widehat{\Sigma^+}$ is equidivisible \cite{Almeida&ACosta:2007a}; see \cite{Almeida&ACosta:2016a} for more general results on equidivisibility of relatively free profinite semigroups.

Given $u$ and $v$ in $\widehat{\Sigma^+}$, we say that a finite semigroup $S$ satisfies the \emph{pseudoidentity} $u=v$ if for every mapping $\varphi:\Sigma \to S$, the equality $\hat{\varphi}(u)=\hat{\varphi}(v)$ holds. 
Since $\widehat{\Sigma^+}$ is an inverse limit of finite semigroups, for every $u$ in $\widehat{\Sigma^+}$ the sequence $(u^{n!})_{n}$ converges to an idempotent, which is denoted $u^\omega$.
Following the same reasoning, for every $u$ in $\widehat{\Sigma^+}$ the sequence $(u^{(n+1)!-1})_{n}$ has a limit, which is denoted $u^{\omega-1}$ and is such that $u^\omega=u^{\omega-1}u$.
Formal equalities between terms built from the letters in $\Sigma$ using only the multiplication and the $\omega$ power are examples of pseudoidentities and will be sufficient for our purposes. 

For a set $\Pi$ of pseudoidentities, we denote by $\llbracket \Pi \rrbracket$  the class of all finite semigroups or monoids (the context should make clear which) that satisfy all pseudoidentities in~$\Pi$. The following pseudovarieties of monoids play a key role in this paper:
\begin{align*}
\mathsf{R} &= \llbracket (xy)^{\omega}x=(xy)^{\omega} \rrbracket
\\\mathsf{L} &= \llbracket y(xy)^{\omega}=(xy)^{\omega} \rrbracket
\\\mathsf{J} &= \mathsf{R} \cap \mathsf{L} = \llbracket (xy)^{\omega}=(yx)^{\omega}, x^{\omega+1}=x^{\omega} \rrbracket
\\\mathsf{DA} &= \llbracket (xy)^{\omega}(yx)^{\omega}(xy)^{\omega}=(xy)^{\omega}, x^{\omega+1}=x^{\omega} \rrbracket
\\&=\llbracket (xy)^{\omega}x(xy)^{\omega}=(xy)^{\omega} \rrbracket.
\end{align*}
We will also refer to the following pseudovarieties of semigroups:
\begin{align*}
\mathsf{D} &= \llbracket yx^{\omega}=x^{\omega} \rrbracket, \\
\mathsf{K} &= \llbracket x^{\omega}y=x^{\omega} \rrbracket.
\end{align*}

For a set $A$, the set ${\cal P} (A)$ of all subsets of $A$ is a monoid under the operation of union, which is the \emph{free semilattice on $A$}. By an \emph{$A$-generated semigroup} we mean a semigroup $S$ endowed with a function $\varphi: A \to S$ such that $S$ is generated by $\varphi(A)$.
We say that the $A$-generated semigroup $S$ has a \emph{content function $c$} if $c: S\to {\cal P} (A)$ is a homomorphism such that $c(\varphi(a))=\{a\}$ for every $a$ in $A$. Abusing notation, all content functions will be denoted $c$. 
In particular, consider the content function for the free profinite semigroup $\widehat{A^+}$ which is the unique continuous extension $c$ of the function $A\to {\cal P}(A)$ sending $a$ to $\{a\}$.

The Trotter-Weil hierarchy of pseudovarieties of monoids is defined by:
\begin{align*}
\mathsf{R}_{1} &= \mathsf{L}_{1} = \mathsf{J} \\
\mathsf{R}_{m+1} &= \mathsf{K} \malcev \mathsf{L}_{m} \\
\mathsf{L}_{m+1} &=  \mathsf{D} \malcev \mathsf{R}_{m}.
\end{align*}
It can be proved that $\mathsf{R}_{2}=\mathsf{R}$ and $\mathsf{L}_{2}=\mathsf{L}$, $\mathsf{R}_{m}\vee \mathsf{L}_{m}\subseteq \mathsf{R}_{m+1}\cap \mathsf{L}_{m+1}$, and $\mathsf{R}_{m}\wedge \mathsf{L}_{m}=\mathsf{R}_{m}\cap \mathsf{L}_{m}$, where join and meet are taken in the lattice of pseudovarieties of monoids. The sublattice of the Trotter-Weil hierarchy is depicted in the following diagram (Figure \ref{diag}).

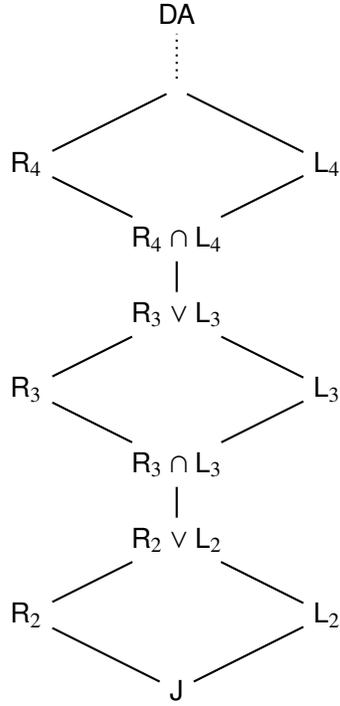
\begin{figure}[H]
	\begin{tikzpicture}	
	\node (J) at (0,0) {$\mathsf{J}$};
	\node (R2) at (-2,1) {$\mathsf{R}_2$};
	\node (L2) at (2,1) {$\mathsf{L}_2$};
	\node (join2) at (0,2) {$\mathsf{R}_{2}\vee \mathsf{L}_{2}$};
	\node (meet3) at (0,3) {$\mathsf{R}_{3}\cap \mathsf{L}_{3}$};
	\node (R3) at (-2,4) {$\mathsf{R}_3$};
	\node (L3) at (2,4) {$\mathsf{L}_3$};
	\node (join3) at (0,5) {$\mathsf{R}_{3}\vee \mathsf{L}_{3}$};
	\node (meet4) at (0,6) {$\mathsf{R}_{4}\cap \mathsf{L}_{4}$};
	\node (R4) at (-2,7) {$\mathsf{R}_4$};
	\node (L4) at (2,7) {$\mathsf{L}_4$};
	\node (empty) at (0,8) {};
	\node (DA) at (0,9) {$\mathsf{DA}$};
	
	\begin{scope}[thick]
	\draw (J) -- (R2) -- (join2)
	(J) -- (L2) -- (join2)
	(join2) -- (meet3)
	(meet3) -- (R3) -- (join3)
	(meet3) -- (L3) -- (join3)
	(join3) -- (meet4)
	(meet4) -- (R4) -- (empty)
	(meet4) -- (L4) -- (empty);
	\draw[dotted](empty) -- (DA);
	\end{scope}
	\end{tikzpicture}
	\centering
	\caption{Diagram of the Trotter-Weil hierarchy}
	\label{diag}
\end{figure}

Following Kufleitner~\cite{Kufleitner:2013}, we define the terms $U_{m},V_{m}$ in the variables $s,t$, $x_{1},\ldots,x_{m},y_{1},\ldots,y_{m}$. For $m\geq 2$, we let
\begin{align*}
U_{1} &= (sx_{1})^{\omega} s(y_{1}t)^{\omega}\\
V_{1} &= (sx_{1})^{\omega} t(y_{1}t)^{\omega}\\
U_{m} &= (U_{m-1}x_{m})^{\omega} U_{m-1} (y_{m}U_{m-1})^{\omega}\\
V_{m} &= (U_{m-1}x_{m})^{\omega} V_{m-1} (y_{m}U_{m-1})^{\omega}.
\end{align*}

\begin{thm}[\cite{Kufleitner:2013}]
	For all $m>1$, we have the following descriptions of pseudoidentities by
	\begin{align*}
	\mathsf{J}=\mathsf{R}_{1} &= \mathsf{L}_{1} = \llbracket U_{1} = V_{1} \rrbracket \\
	\mathsf{R}_{m} \cap \mathsf{L}_{m} &= \llbracket U_{m-1} = V_{m-1} \rrbracket \\
	\mathsf{R}_{m} &= \llbracket (U_{m-1}x_{m})^{\omega} U_{m-1} = (U_{m-1}x_{m})^{\omega}V_{m-1} \rrbracket \\
	\mathsf{L}_{m} &= \llbracket U_{m-1}(y_{m}U_{m-1})^{\omega} = V_{m-1}(y_{m}U_{m-1})^{\omega} \rrbracket . 
	\end{align*}
	\label{Kufleitner,UV}
\end{thm}

We further define the following terms $P_{m},Q_{m}$ in the variables $e,f,s,t, x_{1},\ldots,x_{m}$, $y_{1},\ldots$,$y_{m}$. For $m\geq 2$, we let
\begin{align*}
P_{1} &= (e^\omega sf^\omega x_{1})^{\omega} e^{\omega}sf^{\omega} (y_{1}e^{\omega}tf^{\omega})^{\omega}\\
Q_{1} &= (e^\omega sf^\omega x_{1})^{\omega} e^{\omega}tf^{\omega} (y_{1}e^{\omega}tf^{\omega})^{\omega}\\
P_{m} &= (P_{m-1}x_{m})^{\omega} P_{m-1} (y_{m}P_{m-1})^{\omega}\\
Q_{m} &= (P_{m-1}x_{m})^{\omega} Q_{m-1} (y_{m}P_{m-1})^{\omega}.
\end{align*}

\begin{thm}[\cite{Kufleitner:2013}]
	For all $m\geq 1$, the following equality holds: $$(\mathsf{R}_{m+1} \cap \mathsf{L}_{m+1})*\mathsf{D} = \llbracket P_{m} = Q_{m} \rrbracket.$$
\end{thm}

\section{The pseudovariety $ \mathsf{R}_{m} \cap \mathsf{L}_{m}$ is not local}
Our main result is the inequality $ (\mathsf{R}_{m} \cap \mathsf{L}_{m})*\mathsf{D} \neq L (\mathsf{R}_{m} \cap \mathsf{L}_{m})$, which we establish by giving an example of a language whose syntactic semigroup is in $L (\mathsf{R}_{m} \cap \mathsf{L}_{m}) \setminus (\mathsf{R}_{m} \cap \mathsf{L}_{m})*\mathsf{D}$.
Since, $\mathsf{R}_{1} \cap \mathsf{L}_{1} = \mathsf{R}_{2} \cap \mathsf{L}_{2} = \mathsf{J}$, the cases $m=1,2$ are given by the following well-known theorem. 

\begin{thm}[Knast \cite{Knast:1983a}]
	The pseudovariety $\mathsf{J}$ is not local.
\end{thm}

To show that $ \mathsf{J}*\mathsf{D} \neq L \mathsf{J}$, Knast gave an example of a language whose syntactic semigroup is in $L \mathsf{J} \setminus \mathsf{J}*\mathsf{D}$.
This language is defined as follows: let $A_{2}=\lbrace a,b,c,d \rbrace $ be the alphabet and $\ell_{2}=\bigl(ab^{+} \cup ac^{+}\bigr)^{*}ab^{+}d\bigl(b^{+}d\cup c^{+}d\bigr)^{*}$ be a language of $A_{2}^{+}$.
We then have that $\operatorname{Synt}(\ell_{2})\in L \mathsf{J} \setminus \mathsf{J}*\mathsf{D}$.

We define the alphabets $A_{m}$ and languages $\ell_{m}$ over $A_{m}^{+}$, for $m> 2$, by
\begin{align*}
	A_{m}&=A_{m-1}\cup\lbrace \bar{x}_{m}, \bar{y}_{m}  \rbrace , \\
	\ell_{m}&=\bigl(A_{m-1}\cup\lbrace \bar{x}_{m} \rbrace \bigr)^{*} \bar{x}_{m} \ell_{m-1} \bar{y}_{m} \bigl(A_{m-1}\cup \lbrace  \bar{y}_{m} \rbrace\bigr)^{*},
\end{align*}
where $\bar{x}_{m}$ and $\bar{y}_{m}$ are distinct letters not in $A_{m-1}$.

\begin{prop}
	For all $m\geq 2$, the semigroup $\operatorname{Synt}(\ell_{m}) $ does not belong to $(\mathsf{R}_{m} \cap \mathsf{L}_{m}) *\mathsf{D}$. In other words, $\operatorname{Synt}(\ell_{m}) $ does not satisfy the pseudoidentity $P_{m-1}=Q_{m-1}$. \label{easy-half}
\end{prop}
\begin{proof}
	Define $\Sigma_{1}=\lbrace e,f,s,t,x_{1},y_{1}  \rbrace $ and $\Sigma_{m}=\Sigma_{m-1} \cup \lbrace x_{m}, y_{m} \rbrace  $, where $x_{m}$ and $y_m$ are two distinct letters not in $\Sigma_{m-1}$.
	For each $m\geq2$, we show that there is a continuous homomorphism $\hat\phi_{m}:\widehat{\Sigma_{m-1}^{+}}arrow \widehat{A_{m}^{+}}$, such that $\hat\phi_{m}\bigl(P_{m-1}\bigr) \widensim_{\overline{\ell_{m}}}\hat\phi_{m}\bigl(Q_{m-1}\bigr) $.
	We define $\phi_{2}$ on $\Sigma_{1}$ by letting
	\begin{alignat*}{6}
		 \phi_{2}(e) &=b &\qquad & &  \phi_{2}(s) &= \phi_{2}(x_{1}) =a \\
		  \phi_{2}(f) &=c && & \phi_{2}l(t) &= \phi_{2}(y_{1}) =d.
	\end{alignat*}
	Then we have
	\begin{align*}
		& a\hat\phi_{2}(P_{1}) d=a\bigl(b^{\omega}ac^{\omega}a\bigr) ^{\omega}b^\omega ac^\omega \bigl(db^{\omega}dc^{\omega}\bigr) ^{\omega}d \\
		& a\hat\phi_{2}(Q_{1}) d=a\Bigl(b^{\omega}ac^{\omega}a\Bigr) ^{\omega}b^\omega dc^\omega \bigl(db^{\omega}dc^{\omega}\bigr) ^{\omega}d. 
	\end{align*}
	Since, by equidivisibility of $\widehat{A_2^+}$, for every $\lambda$ in $\overline{b^{+}}$ the pseudoword  $a\lambda d$ is not a factor of $a\hat\phi_{2}(P_{1}) d$, we deduce that $a\hat\phi_{2}(P_{1}) d$ is not an element of $$ \overline{\ell_{2}}= \overline{\bigl(ab^{+} \cup ac^{+}\bigr) ^{*}}\,a\,\overline{b^{+}}\,d \,\overline{\bigl(b^{+}d\cup c^{+}d\bigr) ^{*}}.$$
	Also, we note that
	\begin{align*}
		a\hat\phi_{2}(Q_{1}) d &= \bigl(ab^{\omega}ac^{\omega}\bigr) ^{\omega}ab^\omega dc^\omega d\bigl(b^{\omega}dc^{\omega}d\bigr) ^{\omega} ,
	\end{align*}
	where $\bigl(ab^{\omega}ac^{\omega}\bigr) ^{\omega}$ is an element of $\overline{\bigl(ab^{+} \cup ac^{+}\bigr) ^{*}}$ and $c^\omega d\bigl(b^{\omega}dc^{\omega}d\bigr) ^{\omega}$ is an element of $\overline{\bigl(b^{+}d\cup c^{+}d\bigr) ^{*}}$, so we can say that $a\hat\phi_{2}(Q_{1}) d$ is a pseudoword of $\overline{\ell_{2}}$.
	Since we have $a\hat\phi_{2}(P_{1}) d\notin \overline{\ell_{2}}$ and $a\hat\phi_{2}(Q_{1}) d\in \overline{\ell_{2}}$, we infer that $\hat\phi_{2}(P_{1}) $ and $\hat\phi_{2}(Q_{1}) $ are not syntactically congruent with respect to the set $\overline{\ell_{2}}$.
	
	We now define $\phi_{3}$ on $\Sigma_{2}$ by
	\begin{align*}
		 \phi_{3}|_{\Sigma_{1}} &=\phi_{2} \\
		 \phi_{3}(x_{2})  &=\bar{x}_{3}a \\
		 \phi_{3}(y_{2})  &=d\bar{y}_{3}.
	\end{align*}
	It follows that
	\begin{align*}
		&\hat\phi_{3}(P_{2}) =\Bigl(\hat\phi_{2}(P_{1}) \bar{x}_{3}a\Bigr) ^{\omega}\hat\phi_{2}(P_{1}) \Bigl(d\bar{y}_{3}\hat\phi_{2}(P_{1}) \Bigr) ^{\omega} , \\
		&\hat\phi_{3}(Q_{2}) =\Bigl(\hat\phi_{2}(P_{1}) \bar{x}_{3}a\Bigr) ^{\omega}\hat\phi_{2}(Q_{1}) \Bigl(d\bar{y}_{3}\hat\phi_{2}(P_{1}) \Bigr) ^{\omega}.
	\end{align*}
	Since, for every $\lambda$ in $\overline{\ell_{2}}$, $\bar{x}_{3}\lambda \bar{y}_{3}$ is not a factor of the pseudoword $\hat\phi_{3}(P_{2}) $, we deduce that $\hat\phi_{3}(P_{2}) $ is not an element of $\overline{\ell_{3}}$.
	Also, we note that 
	\begin{align*}
		\hat\phi_{3}(Q_{2}) &= \Bigl(\hat\phi_{2}(P_{1}) \bar{x}_{3}a\Bigr) ^{\omega-1}\hat\phi_{2}(P_{1}) \bar{x}_{3}a\hat\phi_{2}(Q_{1}) d\bar{y}_{3}\hat\phi_{2}(P_{1}) \Bigl(d\bar{y}_{3}\hat\phi_{2}(P_{1}) \Bigr) ^{\omega-1},
	\end{align*}
	where $\Bigl(\hat\phi_{2}(P_{1}) \bar{x}_{3}a\Bigr) ^{\omega-1}\hat\phi_{2}(P_{1}) $ is an element of $\overline{\bigl(A_{2}\cup\lbrace \bar{x}_{3} \rbrace\bigr) ^{*}} $, $\hat\phi_{2}(P_{1}) \Bigl(d\bar{y}_{3}\hat\phi_{2}(P_{1}) \Bigr) ^{\omega-1}$ is an element of $\overline{\bigl(A_{2}\cup\lbrace \bar{y}_{3} \rbrace\bigr) ^{*}}$, and $a\hat\phi_{2}(Q_{1}) d$ is an element of $\overline{\ell_{2}}$, so we can say that $\hat\phi_{3}\bigl(Q_{2}\bigr) $ is a word of $\overline{\ell_{3}}$.
	Since $\hat\phi_{3}(P_{2})  \notin \overline{\ell_{3}}$ and $\hat\phi_{3}(Q_{2})  \in \overline{\ell_{3}}$, it follows that $\hat\phi_{3}(P_{2}) $ and $\hat\phi_{3}(Q_{2}) $ are not syntactically congruent with respect to the set $\overline{\ell_{3}}$.
	
	For $m\geq4$, we now define $\phi_{m}$ on $\Sigma_{m-1}$ by
	\begin{align*}
		 \phi_{m}|_{\Sigma_{m-2}}&=\phi_{m-1} \\
		 \phi_{m}(x_{m-1}) &=\bar{x}_{m}\\ 
		 \phi_{m}(y_{m-1}) &=\bar{y}_{m}.
	\end{align*}
	Suppose, inductively, that $ \hat\phi_{m}(P_{m-2}) =\hat\phi_{m-1}(P_{m-2})  $ is not a pseudoword of $ \overline{\ell_{m-1}}$ and $\hat\phi_{m}(Q_{m-2}) =\hat\phi_{m-1}(Q_{m-2})  $ is an element of  $\overline{\ell_{m-1}}$.
	Note that
	\begin{align*}
		&\hat\phi_{m}(P_{m-1}) =\Bigl(\hat\phi_{m-1}(P_{m-2}) \bar{x}_{m}\Bigr) ^{\omega}\hat\phi_{m-1}(P_{m-2}) \Bigl(\bar{y}_{m}\hat\phi_{m-1}(P_{m-2}) \Bigr) ^{\omega}, \\
		&\hat\phi_{m}(Q_{m-1}) =\Bigl(\hat\phi_{m-1}(P_{m-2}) \bar{x}_{m}\Bigr) ^{\omega}\hat\phi_{m-1}(Q_{m-2}) \Bigl(\bar{y}_{m}\hat\phi_{m-1}(P_{m-2}) \Bigr) ^{\omega}.
	\end{align*}
	Since, for every $\lambda$ in $\overline{\ell_{m-1}}$, $\bar{x}_{m}\lambda \bar{y}_{m}$ is not a factor of the pseudoword $\hat\phi_{m}(P_{m-1}) $, we deduce that $\hat\phi_{m}(P_{m-1}) $ is not pseudoword of $\overline{\ell_{m}}$. Also, we note that 
	\begin{align*}
		\hat\phi_{m}(Q_{m-1})  &= \Bigl(\hat\phi_{m-1}(P_{m-2}) \bar{x}_{m}\Bigr) ^{\omega-1}\hat\phi_{m-1}(P_{m-2}) \cdot\bar{x}_{m}\hat\phi_{m-1}(Q_{m-2}) \bar{y}_{m} \cdot \\
		& \qquad \cdot \hat\phi_{m-1}(P_{m-2}) \Bigl(\bar{y}_{m}\hat\phi_{m-1}(P_{m-2}) \Bigr) ^{\omega-1},
	\end{align*}
	where the pseudoword $\Bigl(\hat\phi_{m-1}(P_{m-2}) \bar{x}_{m}\Bigr) ^{\omega-1}\hat\phi_{m-1}(P_{m-2}) $ belongs to $\overline{\bigl(A_{m-1}\cup\lbrace \bar{x}_{m} \rbrace\bigr) ^{*}}$, $\hat\phi_{m-1}(Q_{m-2}) $ is an element of $\overline{\ell_{m-1}}$, and $\hat\phi_{m-1}(P_{m-2}) \Bigl(\bar{y}_{m}\hat\phi_{m-1}(P_{m-2}) \Bigr) ^{\omega-1} $ is an element of $\overline{\bigl(A_{m-1}\cup\lbrace \bar{y}_{m} \rbrace\bigr) ^{*}}$, so that $\hat\phi_{m}(Q_{m-1}) $ is a pseudoword of $\overline{\ell_{m}}$.
	Since $\hat\phi_{m}(P_{m-1}) $ is not an element $\overline{\ell_{m}}$ and $\hat\phi_{m}(Q_{m-1}) $ is in $\overline{\ell_{m}}$ we deduce that $\hat\phi_{m}(P_{m-1}) $ and $\hat\phi_{m}(Q_{m-1}) $ are not syntactically congruent with respect to the set $\overline{\ell_{m}}$.
	
	As a consequence of what we have just shown, we conclude that there exists a homomorphism  $\psi_{m}:\Sigma_{m}^{+}\rightarrow \operatorname{Synt}\bigl(\ell_{m}\bigr) $ such that $\hat\psi_{m}(P_{m-1}) \neq \hat\psi_{m}(Q_{m-1}) $. That is, the semigroup $\operatorname{Synt}(\ell_{m}) $ does not satisfy the pseudoidentity $P_{m-1}=Q_{m-1}$. We have thus established that $\operatorname{Synt}(\ell_{m}) \notin (\mathsf{R}_{m} \cap \mathsf{L}_{m}) *\mathsf{D}$.
\end{proof}

\begin{prop}
	For all $m\geq 2$, the semigroup $\operatorname{Synt}(\ell_{m})$ belongs to the pseudovariety $L(\mathsf{R}_{m} \cap \mathsf{L}_{m})$. That is, all local monoids of\/ $\operatorname{Synt}(\ell_{m})$ satisfy the pseudoidentity $ U_{m-1}=V_{m-1} $. \label{hard-half}
\end{prop}
The technical proof of Proposition \ref{hard-half} is postponed to Section \ref{hard-proof}.

\begin{thm}
	The pseudovariety $\mathsf{R}_{m} \cap \mathsf{L}_{m}$ is not local for every $m\geq 1$.
\end{thm}
\begin{proof}
	The result follows immediately from Propositions \ref{easy-half} and \ref{hard-half}.
\end{proof}

\section{Proof of Proposition \ref{hard-half}}

Let $\psi_{m}:A_{m}^{+}\rightarrow \operatorname{Synt}(\ell_{m})$ be the natural homomorphism.
For an idempotent $w_{m}$ of $\operatorname{Synt}(\ell_{m})$, consider the local submonoid $w_{m}\operatorname{Synt}(\ell_{m})w_{m}$ of $\operatorname{Synt}(\ell_{m})$.
We take $B_{m}=\hat\psi_{m}^{-1}(w_{m}\operatorname{Synt}(\ell_{m})w_{m})\subseteq \widehat{A_{m}^{+}}$ and we observe that $B_{m}$ is a closed semigroup.

Note that, whenever $u$ is in $\hat\psi_{m}^{-1}(w_{m})$ and $v$ is in $B_{m}$, we have $v\widesim_{\overline{\ell_{m}}}uvu$.

\begin{lem}
	Let $u$ be an element of $\hat\psi_{m}^{-1}(w_{m})$ and $v$ an element of $B_m$. If $\bar{x}_{m}$ (respectively $\bar y_m$) is a letter of $u$, then it is also a letter of $v$.
\end{lem}
\begin{proof}
	Suppose there is $v\in B_{m}$ such that $\bar{x}_{m}$ is not a letter of $v$.
	Note that, for $\lambda \in \overline{\ell_{m-1}}$, we get $\bar{x}_{m}\lambda\bar{y}_{m}v\in \overline{\ell_{m}}$ and $\bar{x}_{m}\lambda\bar{y}_{m}uvu \notin \overline{\ell_{m}}$, which is in contradiction with $v\widesim_{\overline{\ell_{m}}}uvu$.
	Hence for all $v\in B_{m}$, $\bar{x}_{m}$ is a letter of $v$.
	The argument for $\bar y_m$ is similar.
\end{proof}

Let $\Sigma_{1}=\lbrace s,t,x_{1},y_{1}\rbrace $, $\Sigma_{m}=\Sigma_{m-1} \cup \lbrace x_{m}, y_{m} \rbrace$ ($m\geq 2$), and let  $\gamma:\Sigma_{m-1}^{+}\rightarrow B_{m}$ be a semigroup homomorphism.

The next lemma involves the pseudowords $U_m$ and $V_m$ of Theorem \ref{Kufleitner,UV}.

\begin{lem}
	For every $n<m$, the pseudowords $\hat\gamma(U_{n})$ and $\hat\gamma(V_{n})$ have the same content.
\end{lem}
\begin{proof}
	It suffices to show that $U_n$ and $V_n$ have the same content. In fact, it is easy to show inductively that $c(U_n)=\Sigma_n=c(V_n)$. 
\end{proof}
\begin{cor}
	Suppose that $\bar{x}_{m}$ and $\bar{y}_{m}$ are not letters of $\hat\gamma(U_{m-2})$. Then, for every $\sigma$ in $\widehat{\Sigma^{+}_{m-2}}$, $\bar{x}_{m}$ and $\bar{y}_{m}$ are not letters of $\hat\gamma(\sigma)$.
\end{cor}
\begin{proof}
	This is immediate from $c(U_m)=\Sigma_m=c(V_m)$. 
\end{proof}
\begin{lem}
	Let $m>2$. Then, $\operatorname{Synt}(\ell_{m-1})$ is a subsemigroup of $\operatorname{Synt}(\ell_{m})$.
\end{lem}
\begin{proof}
	Let $i:A^{+}_{m-1}\rightarrow A^{+}_{m}$ be the inclusion function. Define $\lambda:\operatorname{Synt}(\ell_{m-1})\rightarrow\operatorname{Synt}(\ell_{m})$ by $\lambda=\psi_{m}\circ i\circ \psi_{m-1}^{-1}$.
	
	\begin{equation*}
		\begin{tikzcd}
			A^{+}_{m-1} \arrow[d,"\psi_{m-1}"'] \arrow[r,"i" ] & A^{+}_{m} \arrow[d,"\psi_{m}"] \\
			\operatorname{Synt}(\ell_{m-1}) \arrow[r,dashed,"\lambda" ] &  \operatorname{Synt}(\ell_{m})
		\end{tikzcd}
	\end{equation*}
	We want to prove that $\lambda$ is a function. Suppose $u$ is an element of $\operatorname{Synt}(\ell_{m-1})$, and take $k$ and $l$ as elements of $\psi_{m-1}^{-1}(u)$. Since $k$ and $l$ are in $\psi_{m-1}^{-1}(u)$, we have that for all $a,b \in A_{m-1}^{*}$, $akb\in \ell_{m-1}$ if and only if $alb\in \ell_{m-1}$.
	We claim that for all $a,b\in A_{m}^{*}$, $akb\in \ell_{m}$ if and only if $alb\in \ell_{m}$.
	
	Suppose that $akb\in \ell_{m}$, therefore $akb=a'\bar x_m b'\bar y_mc'$ where $a',b',c'$ are words of $A_m^*$, $\bar x_m$ is not a letter of $b'$ or $c'$, $\bar y_m$ is not a letter of $a'$ or $b'$, and $b$ is a word of $\ell_m$.
	We consider the three possible situations. Firstly, if $a=a'\bar x_m b'\bar y_m\mu$, where $c'=\mu kb$ and $\mu \in (A_{m-1}\cup \{\bar y_m\})^*$, then $alb=a'\bar x_m b'\bar y_m\mu lb$ and $a'\in (A_{m-1}\cup{\bar x_m})^*$, $b'$ is a word of $\ell_{m-1}$, and $\mu lb \in (A_{m-1}\cup{\bar y_m})^*$; we conclude that $alb$ is a word of $\ell_m$. Secondly, if $c=\mu \bar x_m b'\bar y_m c'$, where $a'=ak \mu$ and $\mu \in (A_{m-1}\cup \{\bar x_m\})^*$, then $alb=al \mu \bar x_m b'\bar y_mc' $ and $al \mu \in (A_{m-1}\cup{\bar x_m})^*$, $b'$ is a word of $\ell_{m-1}$ , and $c' \in (A_{m-1}\cup{\bar y_m})^*$; we conclude that $alb$ is a word of $\ell_m$. Lastly, if $a=a'\bar x_m \mu$, $\mu \in (A_{m-1}\cup \{\bar x_m\})^*$, $c=\mu'\bar y_mc'$ and $\mu' \in (A_{m-1}\cup \{\bar y_m\})^*$, where $b'=\mu k\mu'$, then $alb=a'\bar x_m \mu l \mu'\bar y_mc'$ and $a'\in (A_{m-1}\cup{\bar x_m})^*$, $\mu l \mu' $ is a word of $\ell_{m-1}$, as $k$ and $l$ are syntactically equivalent with respect to $\ell_{m-1}$ , and $c' \in (A_{m-1}\cup{\bar y_m})^*$; we conclude that $alb$ is a word of $\ell_m$. Analogously, we have the converse implication.
	This proves the claim.
	
	Therefore $\psi_{m}\circ i\circ \psi_{m-1}^{-1}(x)$ does not depend on the choice of representative for $\psi_{m-1}^{-1}(x)$. Hence, $\lambda$ defines a function.
	
	We now prove that $\lambda$ defines a semigroup homomorphism. Let $x_{1}$ and $x_{2}$ be elements of $\operatorname{Synt}(\ell_{m-1})$. Since $\psi_{m-1}$ is a homomorphism we have that if $k_{1}\in \psi_{m-1}^{-1}(x_{1})$ and $k_{2}\in \psi_{m-1}^{-1}(x_{2})$, then $k_{1}k_{2}\in\psi_{m-1}^{-1}(x_{1}x_{2})$.
	Therefore, we have that $\lambda(x_{1}x_{2})=\psi_{m}(k_{1}k_{2})=\psi_{m}(k_{1})\psi_{m}(k_{2})=\lambda(x_{1})\lambda(x_{2})$. We thus conclude that $\lambda$ is a homomorphism.
	
	Now we show that $\lambda$ is injective. Let $x_{1}$ and $x_{2}$ elements of $\operatorname{Synt}(\ell_{m-1})$, such that $\lambda(x_{1})=\lambda(x_{2})$ , if $k_{1}$ is an element of $\psi_{m-1}^{-1}(x_{1})$ and $k_{2}$ is an element of $\psi_{m-1}^{-1}(x_{2})$ then $\psi_{m}(k_{1})=\psi_{m}(k_{2})$.
	Then, we have that for all $a,b \in A_{m}^{*}$, $ak_{1}b\in \ell_{m}$ if and only if $ak_{2}b\in \ell_{m}$.
	Taking $a=\bar{x}_{m}a'$, where $a'\in A_{m-1}^{*}$, and $b=b'\bar{y}_{m}$, where $b'\in A_{m-1}^{*}$, we can infer that $\bar{x}_{m}a'k_{1}b'\bar{y}_{m}\in \ell_{m}$ if and only if $\bar{x}_{m}a'k_{2}b'\bar{y}_{m}\in \ell_{m}$, and thus $a'k_{1}b'\in \ell_{m-1}$ if and only if $a'k_{2}b'\in \ell_{m-1}$. Hence, $k_1$ and $k_2$ are syntactically equivalent in $A_{m-1}^+$ with respect to $\ell_{m-1}$.
	Since $\lambda:\operatorname{Synt}(\ell_{m-1})\rightarrow\operatorname{Synt}(\ell_{m})$ is an injective homomorphism, $\operatorname{Synt}(\ell_{m-1})$ is indeed a subsemigroup of $\operatorname{Synt}(\ell_{m})$.
\end{proof}
Let $\lambda$ be the homomorphism defined in the previous proof. 

\begin{cor}
	Let $u$ be an element of $\operatorname{Synt}(\ell_{m})$. There is some $k$ in $\psi_{m}^{-1}(u)$, such that $\bar{x}_{m}$ and $\bar{y}_{m}$ are not letters of $k$, if and only if $u$ is an element of $\lambda(\operatorname{Synt}(\ell_{m-1}))$.
\end{cor}
\begin{proof}
	Let $u\in\operatorname{Synt}(\ell_{m})$ and $k\in\psi_{m}^{-1}(u)$, such that $\bar{x}_{m}$ and $\bar{y}_{m}$ are not letters of $k$.
	As $\bar{x}_{m}$ and $\bar{y}_{m}$ are not letters of $k$, $k$ is an element of $A_{m-1}^{+}$.
	Defining $u'=\psi_{m-1}(k)$, which is an element of $\operatorname{Synt}(\ell_{m-1})$. Note that $\lambda(u')=\psi_{m}(k)=u$, and so $u$ is an element of $\lambda(\operatorname{Synt}(\ell_{m-1}))$.
	
	Suppose that $u$ is an element of $\lambda(\operatorname{Synt}(\ell_{m}))$, and take $k\in\psi_{m-1}^{-1}\circ\lambda^{-1}(u)$. In particular, $k$ is an element of $(A_{m-1})^{+}$ and $k\in \psi_{m}^{-1}(u)$, and therefore there is $k$ in $\psi_{m}^{-1}(u)$ such that $\bar{x}_{m}$ and $\bar{y}_{m}$ are not letters of $k$.
\end{proof}
\begin{lem}
	If $u$ is an element of $\hat\psi_{m}^{-1}(w_{m})$, such that $\bar{x}_{m},  \bar{y}_{m}\notin c(u)$, then we have $u\widesim_{\overline{\ell_{m-1}}}uu$. \label{idempotent:lm-1}
\end{lem}
\begin{proof}
	Let $\alpha$ and $\beta$ be arbitrary elements of $\widehat{A_m^+} \cup \{1\}$.
	As $\bar{x}_{m}$ and $\bar{y}_{m}$ are not letters of $u$, we deduce that $\bar{x}_{m}\alpha u \beta \bar{y}_{m}$ is an element of $\overline{\ell_{m}}$ if and only if $\bar{x}_{m}\alpha uu \beta \bar{y}_{m}$ is an element of $\overline{\ell_{m}}$, and therefore $\alpha u \beta$ is an element of $\overline{\ell_{m-1}}$ if and only if $\alpha uu\beta$ is an element of $\overline{\ell_{m-1}}$. We have thus concluded $u$ and $uu$ are syntactically congruent with respect to $\overline{\ell_{m-1}}$.
\end{proof}
\begin{lem}
	Let $m>2$.
	Suppose there is some $u$ in $\hat\psi_{m}^{-1}(w_{m})$ such that $\bar{x}_{m}$ and $\bar{y}_{m}$ are not letters of $u$.
	If $\bar{x}_{m}$ and $\bar{y}_{m}$ are not letters of $\hat\gamma(U_{m-2})$ and $\hat\gamma(V_{m-2})$, and $\operatorname{Synt}(\ell_{m-1})$ is a member of $L(\mathsf{R}_{m-1} \cap \mathsf{L}_{m-1})$, then for all $a,b \in \widehat{A_{m-1}^{+}}\cup\{1\}$, we have that $a\hat\gamma(U_{m-2})b \in \overline{\ell_{m-1}}$ if and only if $a\hat\gamma(V_{m-2})b \in \overline{\ell_{m-1}}$.
	\label{induction}
\end{lem}
\begin{proof}
	By Lemma \ref{idempotent:lm-1} , we have that $u\widesim_{\ell_{m-1}}uu$, which means that $\hat\psi_{m-1}(u)$ is an idempotent of the semigroup $\operatorname{Synt}(\ell_{m-1})$. And, as $\operatorname{Synt}(\ell_{m-1})$ is a subsemigroup of $\operatorname{Synt}(\ell_{m})$, we have that $\lambda\bigl(\hat\psi_{m-1}(u)\bigr)=w_{m}$.
	
	Consider the set $B_{m}'=\hat\psi_{m}^{-1}\Bigl(w_{m}\lambda\bigl(\operatorname{Synt}(\ell_{m-1})\bigr)w_{m}\Bigr)$.
	Since $\bar{x}_{m}$ and $\bar{y}_{m}$ are not letters of $\hat\gamma(U_{m-2})$, we conclude that pseudowords in $\hat\gamma\Bigl(\widehat{\Sigma_{m-2}^{+}}\Bigr)$ contain neither $\bar{x}_{m}$ nor $\bar{y}_{m}$.
	It follows that $\hat\gamma|_{\widehat{\Sigma_{m-2}^{+}}}:\widehat{\Sigma_{m-2}^{+}}\rightarrow B_{m}'$.
	Since $\operatorname{Synt}(\ell_{m-1})$ is an element of $L(\mathsf{R}_{m-1} \cap \mathsf{L}_{m-1})$, by Theorem \ref{Kufleitner,UV} we deduce that $\hat\gamma(U_{m-2})\widesim_{\overline{\ell_{m-1}}}\hat\gamma(V_{m-2})$.
\end{proof}
 
\begin{proof}[Proof of Proposition \ref{hard-half}]
	\label{hard-proof}
 	Define $\Sigma_{1}=\lbrace s,t,x_{1},y_{1}  \rbrace $ and $\Sigma_{m}=\Sigma_{m-1} \cup \lbrace x_{m}, y_{m} \rbrace  $.
 	Knast proved that $\ell_{2}$ is an element of $L(\mathsf{R}_{2} \cap \mathsf{L}_{2})$ \cite{Knast:1983a}.
 	Assume inductively that $\operatorname{Synt}(\ell_{m-1})$ is a member of $L(\mathsf{R}_{m-1} \cap \mathsf{L}_{m-1})$. We prove that $\operatorname{Synt}(\ell_{m})$ is a member of $L(\mathsf{R}_{m} \cap \mathsf{L}_{m})$.
 	
 	Let $\gamma:\widehat{\Sigma_{m}^{+}}\rightarrow B_{m}$ be a continuous semigroup homomorphism. 
 	We claim that $\gamma(U_{m-1})\sim_{\overline{\ell_{m}}}\gamma(V_{m-1})$.
 	To establish the claim, we consider several cases.
 	
 	Suppose that $a\gamma(U_{m-1})b \in \overline{\ell_{m}}$, that is, there exists a factorization $a\gamma(U_{m-1})b = a' \bar x_m b' \bar y_m c'$ with $a' \in \overline{(A_{m-1} \cup \{\bar x_m\})^+}\cup \{1\}$, $c'\in \overline{(A_{m-1} \cup \{\bar y_m\})^+}\cup \{1\}$, and $b' \in \overline{\ell_{m-1}}$.
 	Taking into account the recursive definition of the $U_m$ (Theorem \ref{Kufleitner,UV}), we obtain the following formula:
 	\begin{equation}
 		a\bigl(\gamma(U_{m-2})\gamma(x_{m-1})\bigr)^\omega \gamma(U_{m-2})\bigl(\gamma(y_{m-1})\gamma(U_{m-2})\bigr)^\omega b = 	a\gamma(U_{m-1})b =  a' \bar x_m b' \bar y_m c'.
 		\label{equality_Um}
 	\end{equation}
 	Similarly, from the definition of the $V_m$, we obtain:
 	\begin{equation}
 		a\gamma(V_{m-1})b=a\bigl(\gamma(U_{m-2})\gamma(x_{m-1})\bigr)^\omega \gamma(V_{m-2})\bigl(\gamma(y_{m-1})\gamma(U_{m-2})\bigr)^\omega b.
 		\label{equality_Vm}
  	\end{equation}
 	The various cases come from invoking equidivisibility in the equality \eqref{equality_Um}.
 	
 	Suppose that $\bar x_m, \bar y_m \in c\Bigl(a\bigl(\gamma(U_{m-2})\gamma(x_{m-1})\bigr)^\omega\Bigr)$.
 	By equidivisibility, we infer that $a \bigl(\gamma(U_{m-2})\gamma(x_{m-1})\bigr)^\omega=a' \bar x_m b' \bar y_m c''$, where $\bar x_m$ does not occur in $b', c''$, $\gamma(U_{m-2}), \gamma(y_{m-1}), b$, and $\bar y_m$ does not occur in $a', b'$.
 	We conclude that $\bar x_m $ does not occur in $\gamma(U_{m-2})$, whence also not in $\gamma(V_{m-2})$, and therefore, in view of \eqref{equality_Vm}, $$a\gamma(V_{m-1})b=a' \bar x_m b' \bar y_m c'' \gamma(V_{m-2})\bigl(\gamma(y_{m-1})\gamma(U_{m-2})\bigr)^\omega b \in \overline{\ell_{m}}.$$
 	The case where $\bar x_m$ and $\bar y_m$ belong to $c\Bigl(\bigl(\gamma(y_{m-1})\gamma(U_{m-2})\bigr)^\omega b\Bigr)$ is dual.
 	
 	Assume that $\bar x_m$ is a member of $c\Bigl(a\bigl(\gamma(U_{m-2})\gamma(x_{m-1})\bigr)^\omega\Bigr)$ and $\bar y_m$ is a member of $c\Bigl(\bigl(\gamma(y_{m-1})\gamma(U_{m-2})\bigr)^\omega b\Bigr)$.
 	We deduce, by equidivisibility, that $a\bigl(\gamma(U_{m-2})\gamma(x_{m-1})\bigr)^\omega=a' \bar x_m b_1' $ and $\bigl(\gamma(y_{m-1})\gamma(U_{m-2})\bigr)^\omega b= b_2' \bar y_m c'$, where $\bar x_m $ does not occur in $b_1',b_2',c'$, $\gamma(U_{m-2})$, and $\bar y_m$ does not occur $a',b_1',b_2',\gamma(U_{m-2})$.
 	We, thus, reach the conclusion that $a \gamma(U_{m-1}) b=a' \bar x_m b_1' \gamma(U_{m-2}) b_2' \bar y_m c' $, where $b_1' \gamma(U_{m-2}) b_2' \in \overline{\ell_{m-1}}$.
 	Since $\operatorname{Synt}(\ell_{m-1})$ belongs to $L(\mathsf{R}_{m-1} \cap \mathsf{L}_{m-1})$, by lemma \ref{induction}, it follows that $b_1'\gamma(V_{m-2})b_2' \in \overline{\ell_{m-1}}$ and, thus, in view of \eqref{equality_Vm}, $a\gamma(V_{m-1})b$ is a member of $\overline{\ell_{m}}$.
 	
 	As every continuous homomorphism $\gamma:\widehat{\Sigma_{m}^{+}}\rightarrow B_{m}$ determines a continuous homomorphism $\gamma':\widehat{\Sigma_{m}^{+}}\rightarrow w_m \operatorname{Synt}(\ell_m)w_m$, for any idempotent $w_m$ of $\operatorname{Synt}(\ell_m)$, and since $\gamma(U_{m-1})\sim_{\overline{\ell_m}}\gamma(V_{m-1})$, we conclude that $\gamma'(U_{m-1})=\gamma'(V_{m-1})$. 
 	We have thus proven that $\operatorname{Synt}(\ell_m)$ belongs to $L(\mathsf{R}_{m} \cap \mathsf{L}_{m})$.
\end{proof}

\section{Acknowledgements}
This work was developed in the framework of the project 'Novos Talentos em Matemática' of the Gulbenkian Foundation, whose support is gratefully acknowledged.

I also want to show my gratitude to Manfred Kufleitner for sharing his insights that greatly helped move my research forward.

I would like to thank my tutor, Professor Jorge Almeida, for sharing his pearls of wisdom with me and for his guidance throughout this project, without whose help and advice this research would never had been possible. I am immensely indebted for the assistance he provided in the writing of this paper and for the comments that greatly improved the manuscript.

\bibliographystyle{amsplain}
\bibliography{ref-sgps-utf8}

\end{document}